\documentclass[12pt,article,onecolumn]{IEEEtran}

\usepackage{float}
\usepackage{graphicx}
\usepackage{amssymb}
\usepackage{epsfig}
\usepackage{amsmath}
\usepackage{amsthm}
\usepackage{amsfonts}
\usepackage{setspace}
\usepackage{url}
\usepackage{dsfont}
\usepackage{ifpdf}

\newtheorem{thm}{Theorem}

\newtheorem{lem}{Lemma}

\newtheorem{defn}{Definition}

\makeatletter

\newcommand{\Rmnum}[1]{\expandafter\@slowromancap\romannumeral #1@}
\makeatother

\newcommand{\abs}[1]{\left\vert#1\right\vert}

\newcommand{\norm}[1]{\left\Vert#1\right\Vert}

\newcommand{\pdf}{\textit{pdf} }
\newcommand{\CDF}{\textit{CDF} }

\newcommand{\Ball}{\textnormal{Ball} }
\newcommand{\spann}{\textnormal{span} }

\newcommand{\mbf}{\mathbf}

\begin{document}

\title{The Random Coding Bound Is Tight for the Average Linear Code or Lattice}


\author{
Yuval~Domb,~\IEEEmembership{Student Member,~IEEE,}
Ram~Zamir,~\IEEEmembership{Fellow,~IEEE,}
Meir~Feder,~\IEEEmembership{Fellow,~IEEE}
\thanks{A subset of this work was presented at the IEEE Convention Israel (IEEEI) 2012.}
}


\maketitle

\begin{abstract}
In 1973, Gallager proved that the random-coding bound is exponentially tight for the random code ensemble at all rates, even below expurgation.
This result explained that the random-coding exponent does not achieve the expurgation exponent due to the properties of the random ensemble, irrespective of the utilized bounding technique.
It has been conjectured that this same behavior holds true for a random ensemble of {\em linear} codes.
This conjecture is proved in this paper.
Additionally, it is shown that this property extends to Poltyrev's random-coding exponent for a random ensemble of lattices.
\end{abstract}

\IEEEpeerreviewmaketitle

\section{Introduction}
The error exponent, for a particular channel, is a function describing the exponential decay rate (with increasing block length) of the maximum-likelihood decoding error probability, for any communication rate $R$ below the capacity $C$.
The random-coding exponent is constructed \cite{b:GallagerInformationTheory} by upper bounding the {\em average} of a maximum-likelihood decoder's error probability over a random ensemble of codes, and considering its exponential decay rate.
In general, the randon-coding exponent has two distinct regions, separated by the critical rate $R_{\textnormal{cr}}$:
\begin{enumerate}
\item The straight line region: $0<R<R_{\textnormal{cr}}$
\item The sphere packing region: $R_{\textnormal{cr}}\leq R<C$
\end{enumerate}
The random-coding exponent is tight in the second region.
This is easily shown via its equality to the sphere packing exponent, which is the exponential decay rate of a lower bound on the error probability in that region \cite{b:GallagerInformationTheory}.
In the first region, there exists an expurgation rate $0<R_{\textnormal{ex}}<R_{\textnormal{cr}}$, such that through expurgation of ``bad" codewords, an exponent better than the random-coding exponent is achievable for any rate $0<R< R_{\textnormal{ex}}$ \cite{b:GallagerInformationTheory}.
Naturally, this gives rise to the question, ``Why is the random-coding exponent not tight in the region $0<R<R_{\textnormal{cr}}$? Is it due to the poor performance of the random ensemble at low rates, or perhaps is it due to the upper bounding technique used for its construction?".
The question is answered in Gallager's 1973 paper \cite{j:GallagerTheRandom}, where a lower bound on the average error probability of the random ensemble is shown, whose exponential rate coincides with the random-coding exponent at all rates.
Evidently, the random-coding exponent at low rates is not tight due to the poor performance of the random ensemble rather than a poor bounding technique.

The random-coding exponent, shown for random codes, applies for random {\em linear} codes as well.\footnotemark{}
\footnotetext{Formally, the random linear ensemble can only achieve the random-coding exponent of channels whose exponent is maximized by a uniformly distributed codeword distribution. This is discussed in greater detail in Section \Rmnum{4}.}
This comes from the fact that the derivation of the error probability upper bound requires only pairwise independence between codewords, a property shared by both ensembles.
However, extending Gallager's lower bound at low rates, for the random linear ensemble, is left a challenge, since it requires triple-wise independence between codewords, a property unmet by the random linear ensemble.
Nonetheless, it has been conjectured that the random-coding exponent is tight for the random linear ensemble, at low rates \cite{j:BargRandomCodesMinimum,j:BargOnTheAsymptotic}.
This paper rigorously proves this conjecture using a new lower bound exponent.
Construction of the new exponent is accomplished by determining the exact distribution of codewords conditioned on other codewords for the random linear ensemble, and utilizing de-Caen's lower bound on the probability of a union of events \cite{j:deCaenALowerBound}\footnotemark{}.
\footnotetext{In \cite{j:CohenLowerBoundsOn}, de-Caen's inequality was used to lower bound the error probability of a specific linear code based on its weight enumeration.}

\begin{figure}[h]
\center{\includegraphics[width=0.5\textwidth]{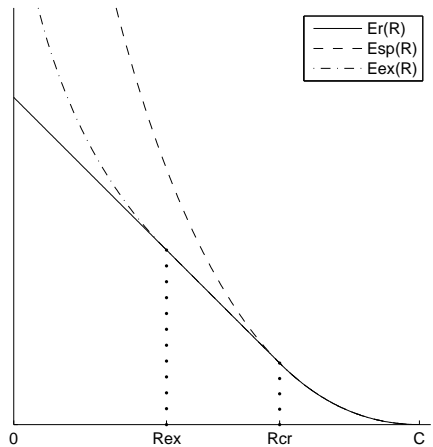}}
\caption{\label{fig_errexp} The random-coding error exponent $E_{\textnormal{r}}(R)$, along with the sphere packing $E_{\textnormal{sp}}(R)$, and expurgation $E_{\textnormal{ex}}(R)$ exponents.}
\end{figure}

Poltyrev's random-coding exponent for unbounded lattices is an exponential upper bound on the average error probability of maximum-likelihood decoding of a lattice point transmitted over an Additive White Gaussian Noise (AWGN) channel, where the average is taken over a \textit{uniformly distributed}\footnotemark{} set of lattices.
\footnotetext{The notion of uniformity for the lattice ensemble is clarified in the paper body.}
To show exponential tightness, our lower bound exponent is extended to the lattice case, by constructing a random lattice ensemble from the random linear ensemble.
This, as turns out, is more complicated than extending the achievability bound \cite{j:LoeligerAveragingBounds} to the lattice case, requiring stringent conditions on the order of the limits in the construction.

The paper is organized as follows:
Section \Rmnum{2} revisits Gallager's proof for the random ensemble.
Section \Rmnum{3} defines the linear random ensemble and establishes its codeword distribution, conditioned on the other codewords.
Section \Rmnum{4} uses the results of Section \Rmnum{3} to construct the new lower bound leading to Theorem \ref{t:thm_lin_tight} which states the tightness of the random-coding exponent for the random linear ensemble.
Section \Rmnum{5} continues with the lattice extension, ending with Theorem \ref{t:thm_lat_tight} which states the tightness of the lattice random-coding exponent for the random lattice ensemble.
Finally, section \Rmnum{6} ends the paper with some concluding remarks.

\section{Gallager's random ensemble revisited}
Let us begin by revisiting Gallager's proof for the random code ensemble using the same terminology and mostly following the same steps as the original.
The differences in the case of the random linear ensemble are highlighted in the next section.
The reader is assumed to be familiar with \cite{b:GallagerInformationTheory,j:GallagerTheRandom}.

The random ensemble is a collection of all codes, such that each code consists of $M=q^K$ codewords, where each codeword is a $q$-ary $N$-tuple, drawn independently from the others, and $K<N$.
The ensemble's average error probability upon transmission of a codeword, corresponding to the $m$'th message where $m\in[0,M-1]$, can be expressed as
\begin{equation}
\label{e:pe_ensemble}
\overline{P}_{e,m}
= \sum_{\mbf{x}_m} \sum_\mbf{y} Q_N(\mbf{x}_m) P_N(\mbf{y}|\mbf{x}_m) \Pr(error|m,\mbf{x}_m,\mbf{y})
\end{equation}
where $\mbf{x}_m$ is the transmitted codeword (corresponding to the $m$'th message), $\mbf{y}$ is the corresponding channel observation, $Q_N(\mbf{x}_m)$ is the a-priori probability for transmitting $\mbf{x}_m$, $P_N(\mbf{y}|\mbf{x}_m)$ is the channel transfer probability, and $\Pr(error|m,\mbf{x}_m,\mbf{y})$ is the probability that upon transmission of $\mbf{x}_m$ and reception of $\mbf{y}$, the decoder decodes $\mbf{x}_{m'}$ with $m'\neq m$, (the probability stems from averaging over the ensemble, since $\mbf{x}_{m'}$ is drawn independently from $\mbf{x}_m$). Maximum-likelihood decoding is implicity assumed throughout the paper.

Suppose $\mbf{x}$ is transmitted and $\mbf{y}$ received, and let $A(\mbf{x},\mbf{y})$ be the set of all possible channel inputs $\mbf{x}'$ that are more likely to be decoded than $\mbf{x}$, thus,
\begin{equation}
\label{e:single_error_event_noind}
A(\mbf{x},\mbf{y})
= \left\{\mbf{x}': \frac{P_N(\mbf{y}|\mbf{x}')}{P_N(\mbf{y}|\mbf{x})} \geq 1 \right\},
\end{equation}
Define $\sigma(\mbf{x}',\mbf{x},\mbf{y})$ to be the characteristic function of $A(\mbf{x},\mbf{y})$, i.e.
\begin{align}
\label{e:single_error_event_noind_char}
\sigma(\mbf{x}',\mbf{x},\mbf{y}) =
\left\{
\begin{array}{ll}
1 & : \mbf{x}' \in A(\mbf{x},\mbf{y}) \\
0 & : \mbf{x}' \notin A(\mbf{x},\mbf{y}).
\end{array}
\right.
\end{align}
Let $A_{m'}(\mbf{x}_m,\mbf{y})$ be the event that for a specific index $m'\neq m$, $\mbf{x}_{m'}$ is such as to cause an error for $\mbf{x}_m$ transmitted and $\mbf{y}$ received; thus\footnotemark{}
\footnotetext{Gallager defined $A_{m'}(\mbf{x}_m,\mbf{y})$ equivalently as $$A_{m'}(\mbf{x}_m,\mbf{y}) = \left\{\mbf{x}_{m'}: \frac{P_N(\mbf{y}|\mbf{x}_{m'})}{P_N(\mbf{y}|\mbf{x}_m)} \geq 1 \right\}$$}
\begin{equation}
\label{e:single_error_event}
A_{m'}(\mbf{x}_m,\mbf{y})
= \{\mbf{x}_{m'}\in A(\mbf{x}=\mbf{x}_m,\mbf{y}) \}.
\end{equation}
The error probability $\Pr(error|m,\mbf{x}_m,\mbf{y})$ can now be expressed as the probability (w.r.t $Q_N(\mbf{x}_{m'})$) of the union of all error events \eqref{e:single_error_event},
\begin{equation}
\label{e:pr_error_m_x_y}
\Pr(error|m,\mbf{x}_m,\mbf{y})
= \Pr \left( \bigcup_{m'\neq m} A_{m'}(\mbf{x}_m,\mbf{y}) \right).
\end{equation}

The expression \eqref{e:pr_error_m_x_y} can be bounded as follows:
Upper bound using the union bound,
\begin{equation}
\label{e:pr_error_m_x_y_up}
\Pr(error|m,\mbf{x}_m,\mbf{y})
\leq \sum_{m'\neq m} \Pr(A_{m'}(\mbf{x}_m,\mbf{y})),
\end{equation}
and lower bound using the Bonferroni inequality \cite{j:deCaenALowerBound},
\begin{align}
\label{e:pr_error_m_x_y_down}
\Pr&(error|m,\mbf{x}_m,\mbf{y})
\geq \sum_{m'\neq m} \Pr(A_{m'}(\mbf{x}_m,\mbf{y}))\nonumber \\
&- 0.5 \cdot \sum_{\substack{m'\neq m\\m''\neq m,m'}} \Pr(A_{m'}(\mbf{x}_m,\mbf{y}) \cap A_{m''}(\mbf{x}_m,\mbf{y})).
\end{align}
In order to highlight the differences from the random linear ensemble, let us proceed in a slightly more detailed fashion than perhaps required.
Using \eqref{e:single_error_event_noind_char}, $\Pr(A_{m'}(\mbf{x}_m,\mbf{y}))$ can be expressed as
\begin{equation}
\label{e:pr_single_error_event}
\Pr(A_{m'}(\mbf{x}_m,\mbf{y})) = \sum_{\mbf{x}_{m'}} Q_N(\mbf{x}_{m'}|\mbf{x}_m) \sigma(\mbf{x}_{m'},\mbf{x}_m,\mbf{y})
\end{equation}
where $Q_N(\mbf{x}_{m'}|\mbf{x}_m)$ is the probability of the codeword $\mbf{x}_{m'}$ conditioned on $\mbf{x}_m$.
$\Pr(A_{m'}(\mbf{x}_m,\mbf{y}) \cap A_{m''}(\mbf{x}_m,\mbf{y}))$ can be similarly expressed as
\begin{align}
\label{e:pr_double_error_event}
\Pr&(A_{m'}(\mbf{x}_m,\mbf{y}) \cap A_{m''}(\mbf{x}_m,\mbf{y})) \nonumber \\
&= \sum_{\mbf{x}_{m'},\mbf{x}_{m''}}Q_N(\mbf{x}_{m'},\mbf{x}_{m''}|\mbf{x}_m) \sigma(\mbf{x}_{m'},\mbf{x}_m,\mbf{y}) \sigma(\mbf{x}_{m''},\mbf{x}_m,\mbf{y})
\end{align}
where $Q_N(\mbf{x}_{m'},\mbf{x}_{m''}|\mbf{x}_m)$ is the probability of the codewords $\mbf{x}_{m''}$ and $\mbf{x}_{m'}$ conditioned on $\mbf{x}_m$.
Since the random ensemble's codewords are uniformly and independently distributed
\begin{align}
\label{e:pr_rnd_1_cond_1&2}
&Q_N(\mbf{x}_{m'}|\mbf{x}_m) = Q_N(\mbf{x}_{m'}) = q^{-N} \nonumber \\
&Q_N(\mbf{x}_{m'},\mbf{x}_{m''}|\mbf{x}_m) = Q_N(\mbf{x}_{m'})Q_N(\mbf{x}_{m''}) = q^{-2N}.
\end{align}
Temporarily suppressing the dependence on $\mbf{x}_m$ and $\mbf{y}$, define
\begin{equation}
\label{e:def_rnd_alpha}
\alpha \triangleq \Pr(A_{m'}(\mbf{x}_m,\mbf{y})).
\end{equation}
Plugging \eqref{e:pr_single_error_event}, \eqref{e:pr_double_error_event}, and \eqref{e:pr_rnd_1_cond_1&2} into \eqref{e:pr_error_m_x_y_up} and \eqref{e:pr_error_m_x_y_down}, while using the definition from \eqref{e:def_rnd_alpha}, yields the following bounds on $\Pr(error|m,\mbf{x}_m,\mbf{y})$, for the random ensemble:
\begin{equation}
\label{e:pr_rnd_error_m_x_y_0}
(M-1)\alpha - [(M-1)\alpha]^2 \leq \Pr(error|m,\mbf{x}_m,\mbf{y}) \leq (M-1)\alpha.
\end{equation}
The second term in the left-hand-side of \eqref{e:pr_rnd_error_m_x_y_0} can be replaced with $[(M-1)\alpha]^{\rho}$ for $1\leq\rho\leq2$ by noting that when $(M-1)\alpha\geq1$, $(M-1)\alpha - [(M-1)\alpha]^{\rho}\leq0$, and when $(M-1)\alpha<1$ the lower bound is only weakened, thus
\begin{equation}
\label{e:pr_rnd_error_m_x_y}
(M-1)\alpha - [(M-1)\alpha]^{\rho} \leq \Pr(error|m,\mbf{x}_m,\mbf{y}) \leq (M-1)\alpha
\end{equation}
where $1\leq\rho\leq2$.
Gallager uses \eqref{e:pr_rnd_error_m_x_y} to show that the upper and lower bounds are asymptotically exponentially equal.
The purpose of this paper is to show that a bounding similar to \eqref{e:pr_rnd_error_m_x_y} is applicable for the random linear ensemble.\footnotemark{}
\footnotetext{In \cite{j:BargOnTheAsymptotic} Barg points out that plugging the term $[(M-1)\alpha]^\rho$ for $\Pr(error|m,\mbf{x}_m,\mbf{y})$ in \eqref{e:pe_ensemble} and maximizing over $1\leq\rho\leq2$ for $R<R_{\textnormal{cr}}$ results in the error exponent for list decoding with a list of size 2.}

\section{The random linear ensemble}
The conditional codeword distribution for the random linear ensemble is unfortunately not as simple as \eqref{e:pr_rnd_1_cond_1&2}.
This section defines the random linear ensemble and calculates its conditional codeword distribution.

\subsection{Definition}
Denote by $\mathbb{F}_q$ the finite field with $q$-ary elements and denote by $\mathbb{F}_q^n$ its $n$-dimensional vector extension.
All vector operations are performed with the usual element-wise scalar mod-$q$.
It is implicitly assumed that all operations hereafter are with regards to this finite vector field.

Define a linear block code $\mathcal{C}$ of rate $R=\frac{K}{N}$ as a $K$-dimensional translated linear subspace of $\mathbb{F}_q^N$.
More specifically, $\mathcal{C}$ is defined as the collection of all codewords resulting from the linear mapping
\begin{equation}
\label{e:def_lin_code}
\mathcal{C} = \{f: \mathbb{F}_q^K \mapsto \mathbb{F}_q^N, f(\mbf{u})=\mbf{u}G+\mbf{v}\}
\end{equation}
where $G$ and $\mbf{v}$ are some $K\times N$ matrix and $1\times N$ vector of $q$-ary coefficients, respectively.
To simplify the notations, all vectors are assumed to be row vectors, throughout the paper.

Define the ensemble of all linear block codes of rate $R=\frac{K}{N}$ as the collection of all codes \eqref{e:def_lin_code} with the elements of $G$ and $\mbf{v}$ uniformly distributed and independent from each other.
Unless otherwise indicated, all future distributions are calculated over the random linear ensemble, which is conveniently referred to as the linear ensemble.

For consistency with Gallager's terminology, a codeword $\mbf{x}_m$ of index $m$ is implicity assumed to result from the transformation $\mbf{x}_m=\mbf{u}_mG+\mbf{v}$ where $\mbf{u}_m$ is the $q$-ary representation of the index $m$.

\subsection{The pairwise conditional codeword distribution}
This section presents a simple analysis of the codeword distribution conditioned on a single other codeword, as a preview to the next section which provides a theorem outlining the general case.

Due to linearity, $Q_N(\mbf{x}_{m'}|\mbf{x}_{m_1})$ can be manipulated as follows:
\begin{align}
\label{e:pr_lin_1_cond_1_1}
Q_N&(\mbf{x}_{m'}|\mbf{x}_{m_1}) \nonumber \\
&= \Pr(\mbf{x}_{m'}=\mbf{u}_{m'}G+\mbf{v}|\mbf{x}_{m_1}=\mbf{u}_{m_1}G+\mbf{v}) \nonumber \\
&= \Pr(\mbf{x}_{m'}-\mbf{x}_{m_1}=(\mbf{u}_{m'}-\mbf{u}_{m_1})G|\mbf{v}=\mbf{x}_{m_1}-\mbf{u}_{m_1}G) \nonumber \\
&= \Pr(\hat{\mbf{x}}=\hat{\mbf{u}}G)
\end{align}
where $\hat{\mbf{x}}\triangleq \mbf{x}_{m'}-\mbf{x}_{m_1}$ and $\hat{\mbf{u}}\triangleq \mbf{u}_{m'}-\mbf{u}_{m_1}$ (note that since $\mbf{u}_{m'}\neq \mbf{u}_{m_1}$, $\hat{\mbf{u}}\in\{1,2,\ldots,q^K-1\}$, i.e. $\hat{\mbf{u}}$ is never zero). The last equality is due to $G$ being independent of $\mbf{v}$.
Finally, for any legal value of $\hat{\mbf{u}}$, $\Pr(\hat{\mbf{x}}=\hat{\mbf{u}}G)$ is uniformly distributed due to the randomness of $G$, thus
\begin{equation}
\label{e:pr_lin_1_cond_1_2}
Q_N(\mbf{x}_{m'}|\mbf{x}_{m_1}) = q^{-N}.
\end{equation}
Equation \eqref{e:pr_lin_1_cond_1_2} is commonly referred to as the pairwise independence of the linear ensemble.
This property is key in upper bounding the linear ensemble's error probability, using the random-coding bound \cite{b:GallagerInformationTheory}.

\subsection{The generalized conditional codeword distribution}
Begin by defining a new operator to simplify the presentation.
\begin{defn}
Let $\textnormal{span}^*(\mbf{x}_1,\mbf{x}_2,\ldots,\mbf{x}_n)$ define the following translated $(n-1)$-dimensional linear subspace,
\begin{equation}
\textnormal{span}^*(\mbf{x}_1,\mbf{x}_2,\ldots,\mbf{x}_n) \triangleq \mbf{x}_1+\textnormal{span}(\mbf{x}_2-\mbf{x}_1,\ldots,\mbf{x}_n-\mbf{x}_1).
\end{equation}
Note that the role of $\mbf{x}_1$ in the definition above can be replaced by any of the other arguments $\mbf{x}_2,\ldots,\mbf{x}_n$.
$\textnormal{span}^*(\mbf{x}_1,\mbf{x}_2,\ldots,\mbf{x}_n)$ is indifferent to and equivalent for any choice of $\mbf{x}_i$.
\end{defn}

The following theorem outlines the conditional codeword distribution in the general case.
\begin{thm}
\label{t:thm_pr_lin_1_cond_k}
Let $\mbf{x}_{m'},\mbf{x}_{m_1},\ldots,\mbf{x}_{m_k}$ be codewords from the linear ensemble where $k\geq 2$, and let $S_\mbf{u} = \{\mbf{u}: \mbf{u}\in \textnormal{span}^*\{\mbf{u}_{m_1},\ldots,\mbf{u}_{m_k}\}\}$, then the ensemble average of the distribution of $\mbf{x}_{m'}$ conditioned on $\mbf{x}_{m_1},\ldots,\mbf{x}_{m_k}$, is given by the following expressions: \\
When $\mbf{u}_{m'} \notin S_\mbf{u}$,
\begin{equation}
\label{e:thm_pr_lin_1_cond_k_u_notinSu}
Q_N(\mbf{x}_{m'}|\mbf{x}_{m_1},\ldots,\mbf{x}_{m_k}) = q^{-N}.
\end{equation}
When $\mbf{u}_{m'} \in S_\mbf{u}$,
\begin{equation}
\label{e:thm_pr_lin_1_cond_k_u_inSu}
Q_N(\mbf{x}_{m'}|\mbf{x}_{m_1},\ldots,\mbf{x}_{m_k})
= \left\{
\begin{array}{ll}
1 & : \mbf{x}_{m'} = \mathcal{L}(\mbf{x}_{m_1},\ldots,\mbf{x}_{m_k}) \\
0 & : \mbf{x}_{m'} \neq \mathcal{L}(\mbf{x}_{m_1},\ldots,\mbf{x}_{m_k})
\end{array}
\right.
\end{equation}
where $\mathcal{L}$ is any linear transformation such that
\begin{equation}
\mbf{u}_{m'} = \mathcal{L}(\mbf{u}_{m_1},\ldots,\mbf{u}_{m_k}).
\end{equation}
\end{thm}
\begin{proof}
Define $\hat{\mbf{u}}\triangleq \mbf{u}_{m'}-\mbf{u}_{m_k}$, $\hat{\mbf{u}}_i\triangleq \mbf{u}_{m_i}-\mbf{u}_{m_k},\forall i=1,\ldots,k-1$, and $\hat{S}_\mbf{u} = \{\mbf{u}: \mbf{u}\in \textnormal{span}\{\hat{\mbf{u}}_1,\ldots,\hat{\mbf{u}}_{k-1}\}\}$ (note that since $\mbf{u}_{m'},\mbf{u}_{m_1},\ldots,\mbf{u}_{m_k}$ are unique, $\hat{\mbf{u}},\hat{\mbf{u}}_1,\ldots,\hat{\mbf{u}}_{k-1}$ are nonzero and unique, thus $\hat{\mbf{u}}\in\{\{1,2,\ldots,q^K-1\}\setminus\{\hat{\mbf{u}}_1,\ldots,\hat{\mbf{u}}_{k-1}\}\}$).
Define $\hat{\mbf{x}}\triangleq \mbf{x}_{m'}-\mbf{x}_{m_k}$ and $\hat{\mbf{x}}_i\triangleq \mbf{x}_{m_i}-\mbf{x}_{m_k},\forall i=1,\ldots,k-1$.
Using these definitions, $Q_N(\mbf{x}_{m'}|\mbf{x}_{m_1},\ldots,\mbf{x}_{m_k})$ can be manipulated in a similar manner to \eqref{e:pr_lin_1_cond_1_1},
\begin{align}
\label{e:pr_lin_1_cond_k_0}
Q_N&(\mbf{x}_{m'}|\mbf{x}_{m_1},\ldots,\mbf{x}_{m_k}) \nonumber \\
&= \Pr(\hat{\mbf{x}}=\hat{\mbf{u}}G|\hat{\mbf{x}}_1=\hat{\mbf{u}}_1G,\ldots,\hat{\mbf{x}}_{k-1}=\hat{\mbf{u}}_{k-1}G).
\end{align}
Define $\{j_i\}_{i=1}^t$ to be a set of indexes such that the vectors in the set $\{\hat{\mbf{u}}_{j_1},\ldots,\hat{\mbf{u}}_{j_t}\}$ are all linearly independent and $\textnormal{span}(\hat{\mbf{u}}_{j_1},\ldots,\hat{\mbf{u}}_{j_t})=\hat{S}_\mbf{u}$ (in other words, $\{\hat{\mbf{u}}_{j_1},\ldots,\hat{\mbf{u}}_{j_t}\}$ is the largest linearly independent subset of $\{\hat{\mbf{u}}_1,\ldots,\hat{\mbf{u}}_{k-1}\}$).
With that, some of the conditions in \eqref{e:pr_lin_1_cond_k_0} can possibly be removed so that
\begin{align}
\label{e:pr_lin_1_cond_k_1}
Q_N&(\mbf{x}_{m'}|\mbf{x}_{m_1},\ldots,\mbf{x}_{m_k}) \nonumber \\
&= \Pr(\hat{\mbf{x}}=\hat{\mbf{u}}G|\hat{\mbf{x}}_{j_1}=\hat{\mbf{u}}_{j_1}G,\ldots,\hat{\mbf{x}}_{j_t}=\hat{\mbf{u}}_{j_t}G).
\end{align}
Due to linearity, the right-hand-side of \eqref{e:pr_lin_1_cond_k_1} can be rewritten as
\begin{align}
\label{e:pr_lin_1_cond_k_2}
\Pr&(\hat{\mbf{x}}=\hat{\mbf{u}}G|\hat{\mbf{x}}_{j_1}=\hat{\mbf{u}}_{j_1}G,\ldots,\hat{\mbf{x}}_{j_t}=\hat{\mbf{u}}_{j_t}G) \nonumber \\
&= \Pr(\hat{\mbf{x}}=\widetilde{\mbf{u}}\widetilde{G}|\hat{\mbf{x}}_{j_1}=\widetilde{\mbf{u}}_{j_1}\widetilde{G},\ldots,\hat{\mbf{x}}_{j_t}=\widetilde{\mbf{u}}_{j_t}\widetilde{G})
\end{align}
where $\widetilde{\mbf{u}}_{j_i}$ is defined as an all zeros vector with a single $1$ at position $i,\forall i=1,\ldots,t$,
\begin{equation}
\widetilde{G} \triangleq
\left(
\begin{array}{c}
\hat{\mbf{x}}_{j_1} \\
\vdots \\
\hat{\mbf{x}}_{j_t} \\
\hline \\
\widetilde{G}'
\end{array}
\right)
\end{equation}
and $\widetilde{G}'$ is a uniformly distributed $(K-t)\times N$ matrix of $q$-ary coefficients.
The transformation in \eqref{e:pr_lin_1_cond_k_2} can be achieved by introducing a full rank $K\times K$ matrix $A$, such that $G=A\widetilde{G}$, $\hat{\mbf{u}}=\widetilde{\mbf{u}} A^{-1}$, and $\hat{\mbf{u}}_{j_i}=\widetilde{\mbf{u}}_{j_i} A^{-1}$ (note that $\widetilde{\mbf{u}}\in\{\{1,2,\ldots,q^K-1\}\setminus\{2^{i-1}:1\leq i\leq t\}\}$.
With the transformation, it is obvious that the conditioning in the right-hand-side of \eqref{e:pr_lin_1_cond_k_2} can be changed to
\begin{equation}
\label{e:pr_lin_1_cond_k_3}
Q_N(\mbf{x}_{m'}|\mbf{x}_{m_1},\ldots,\mbf{x}_{m_k})
= \Pr(\hat{\mbf{x}}=\widetilde{\mbf{u}}\widetilde{G}|\widetilde{\mbf{u}}_{j_1},\ldots,\widetilde{\mbf{u}}_{j_t}).
\end{equation}
There are two cases for $\Pr(\hat{\mbf{x}}=\widetilde{\mbf{u}}\widetilde{G}|\widetilde{\mbf{u}}_{j_1},\ldots,\widetilde{\mbf{u}}_{j_t})$, depending on the value of $\widetilde{\mbf{u}}$, or more specifically its $K-t$ rightmost elements $\widetilde{\mbf{u}}_{t+1:K}$: \\
When $\widetilde{\mbf{u}}_{t+1:K}\neq0$ ($\hat{\mbf{u}}\notin \hat{S}_\mbf{u}$ and therefore $\mbf{u}_{m'}\notin S_\mbf{u}$), then
\begin{equation}
\label{e:pr_lin_1_cond_k_u_notinSu}
\Pr(\hat{\mbf{x}}=\widetilde{\mbf{u}}\widetilde{G}|\widetilde{\mbf{u}}_{j_1},\ldots,\widetilde{\mbf{u}}_{j_t}) = q^{-N}.
\end{equation}
When $\widetilde{\mbf{u}}_{t+1:K}=0$ ($\hat{\mbf{u}}\in \hat{S}_\mbf{u}$ and therefore $\mbf{u}_{m'}\in S_\mbf{u}$), then
\begin{equation}
\label{e:pr_lin_1_cond_k_u_inSu}
\Pr(\hat{\mbf{x}}=\widetilde{\mbf{u}}\widetilde{G}|\widetilde{\mbf{u}}_{j_1},\ldots,\widetilde{\mbf{u}}_{j_t})
= \left\{
\begin{array}{ll}
1 & : \hat{\mbf{x}} = \hat{\mathcal{L}}(\hat{\mbf{x}_1},\ldots,\hat{\mbf{x}_{k-1}}) \\
0 & : \hat{\mbf{x}} \neq \hat{\mathcal{L}}(\hat{\mbf{x}_1},\ldots,\hat{\mbf{x}_{k-1}})
\end{array}
\right.
\end{equation}
where $\hat{\mathcal{L}}$ is the unique linear transformation such that
\begin{equation}
\hat{\mbf{u}} = \hat{\mathcal{L}}(\hat{\mbf{u}}_{j_1},\ldots,\hat{\mbf{u}}_{j_t}).
\end{equation}
The transformation $\mathcal{L}$ can be taken as
\begin{equation}
\mathcal{L}(\mbf{u}_{m_1},\ldots,\mbf{u}_{m_k}) = \hat{\mathcal{L}}(\hat{\mbf{u}}_{j_1},\ldots,\hat{\mbf{u}}_{j_t}) + \mbf{u}_{m_k}.
\end{equation}
Plugging \eqref{e:pr_lin_1_cond_k_u_notinSu} and \eqref{e:pr_lin_1_cond_k_u_inSu} into \eqref{e:pr_lin_1_cond_k_1} completes the proof.
\end{proof}

\section{Bounding $\overline{P}_{e,m}$ for the linear ensemble}
Let us use the results of the previous section, to repeat Gallager's analysis for the linear ensemble.

Before proceeding, let us make a general clarification regarding the linear ensemble.
Due to its linear structure, the ensemble's codeword distribution is uniform.
As such, it can only achieve the error exponent of channels whose exponent is maximized by a uniform input distribution.
Denote by $\mathcal{U}$ the class of all channels obeying the above.
It is assumed hereafter that all channels considered belong to $\mathcal{U}$.\footnotemark{}
\footnotetext{An alternative treatment is possible by examining the exponent resulting from a uniform input distribution regardless of the channel.}

\subsection{The pairwise intersection of error events}
Begin by noting from \eqref{e:pr_lin_1_cond_1_2} that the distribution of a codeword conditioned on a single other codeword is identical for the random and linear ensembles.
Then from \eqref{e:def_rnd_alpha} we can say that also for the linear ensemble
\begin{equation}
\label{e:def_lin_alpha}
\Pr(A_{m'}(\mbf{x}_m,\mbf{y}))=\alpha.
\end{equation}
Unlike before, the distribution of a codeword conditioned on two other codewords differs for the random and linear ensembles, so that the probability of the intersection of two error events is also different.

\begin{lem}
\label{t:lem_lin_double_error_event}
Let $\mbf{u}_m,\mbf{u}_{m'},\mbf{u}_{m''}$ be distinct message vectors, let $\mbf{x}_m$ be some codeword from the linear ensemble such that $q^K>2$, and let $\mbf{y}$ be some channel output, then the ensemble average of the probability of the intersection of two error events is given by the following expressions: \\
When $\mbf{u}_m \notin \textnormal{span}^*\{\mbf{u}_{m'},\mbf{u}_{m''}\}$,
\begin{equation}
\label{e:lem_pr_lin_double_error_event_1}
\Pr(A_{m'}(\mbf{x}_m,\mbf{y})) \cap A_{m''}(\mbf{x}_m,\mbf{y})) = \alpha^2.
\end{equation}
When $\mbf{u}_m \in \textnormal{span}^*\{\mbf{u}_{m'},\mbf{u}_{m''}\}$,
\begin{equation}
\label{e:lem_pr_lin_double_error_event_2}
\Pr(A_{m'}(\mbf{x}_m,\mbf{y})) \cap A_{m''}(\mbf{x}_m,\mbf{y})) \leq \alpha
\end{equation}
where $\alpha\triangleq\Pr(A_{m'}(\mbf{x}_m,\mbf{y}))$. Note that $\abs{\textnormal{span}^*\{\mbf{u}_{m'},\mbf{u}_{m''}\}}=q$.
\end{lem}

\begin{proof}
Begin by manipulating \eqref{e:pr_double_error_event} as follows:\emph{}
\begin{align}
\label{e:pr_lin_double_error_event_1}
\Pr&(A_{m'}(\mbf{x}_m,\mbf{y}) \cap A_{m''}(\mbf{x}_m,\mbf{y})) \nonumber \\
&= \sum_{\substack{\mbf{x}_{m'}\in \mathbb{F}_q^N\\\mbf{x}_{m''}\in \mathbb{F}_q^N}} Q_N(\mbf{x}_{m'},\mbf{x}_{m''}|\mbf{x}_m) \sigma(\mbf{x}_{m'},\mbf{x}_m,\mbf{y}) \sigma(\mbf{x}_{m''},\mbf{x}_m,\mbf{y}) \nonumber \\
&= \sum_{\mbf{x}_{m'}\in \mathbb{F}_q^N} Q_N(\mbf{x}_{m'}|\mbf{x}_m) \sigma(\mbf{x}_{m'},\mbf{x}_m,\mbf{y}) \nonumber \\
&\qquad \cdot \sum_{\mbf{x}_{m''}\in \mathbb{F}_q^N} Q_N(\mbf{x}_{m''}|\mbf{x}_{m'},\mbf{x}_m) \sigma(\mbf{x}_{m''},\mbf{x}_m,\mbf{y}) \nonumber \\
&= q^{-N} \sum_{\mbf{x}_{m'}\in \mathbb{F}_q^N} \sigma(\mbf{x}_{m'},\mbf{x}_m,\mbf{y}) \nonumber \\
&\qquad \cdot \sum_{\mbf{x}_{m''}\in \mathbb{F}_q^N} Q_N(\mbf{x}_{m''}|\mbf{x}_{m'},\mbf{x}_m) \sigma(\mbf{x}_{m''},\mbf{x}_m,\mbf{y})
\end{align}
where the second equality is due to Bayes' Law and the third is due to \eqref{e:pr_lin_1_cond_1_2}.
We continue by analysing two cases: \\
When $\mbf{u}_m \notin \textnormal{span}^*\{\mbf{u}_{m'},\mbf{u}_{m''}\}$ then by \eqref{e:thm_pr_lin_1_cond_k_u_notinSu}, $Q_N(\mbf{x}_{m''}|\mbf{x}_{m'},\mbf{x}_m)=q^{-N}$ and
\begin{equation}
\label{e:pr_lin_double_error_event_2}
\Pr(A_{m'}(\mbf{x}_m,\mbf{y})) \cap A_{m''}(\mbf{x}_m,\mbf{y})) = \alpha^2.
\end{equation}
When $\mbf{u}_m \in \textnormal{span}^*\{\mbf{u}_{m'},\mbf{u}_{m''}\}$ then by \eqref{e:thm_pr_lin_1_cond_k_u_inSu}
\begin{equation}
\label{e:pr_lin_1_cond_2_u_inS}
Q_N(\mbf{x}_m|\mbf{x}_{m'},\mbf{x}_{m''})
= \left\{
\begin{array}{ll}
1 & : \mbf{x}_m = \mathcal{L}(\mbf{x}_{m'},\mbf{x}_{m''}) \\
0 & : \mbf{x}_m \neq \mathcal{L}(\mbf{x}_{m'},\mbf{x}_{m''})
\end{array}
\right.
\end{equation}
where $\mathcal{L}$ is any linear transformation such that
\begin{equation}
\mbf{u}_m = \mathcal{L}(\mbf{u}_{m'},\mbf{u}_{m''}).
\end{equation}
Plugging \eqref{e:pr_lin_1_cond_2_u_inS} into \eqref{e:pr_lin_double_error_event_1} and upper bounding by taking $\sigma(\mbf{x}_{m''},\mbf{x}_m,\mbf{y})=1$ whenever $Q_N(\mbf{x}_m|\mbf{x}_{m'},\mbf{x}_{m''})\neq0$ results in
\begin{equation}
\label{e:pr_lin_double_error_event_3}
\Pr(A_{m'}(\mbf{x}_m,\mbf{y})) \cap A_{m''}(\mbf{x}_m,\mbf{y})) \leq \alpha.
\end{equation}
\end{proof}

\subsection{Bounding the union of error events}
Attempting to use \eqref{e:pr_error_m_x_y_down} to lower bound the union of error events, as was done for the random ensemble, results in a term which is always negative leading to a useless bound, (excluding the case when $q=2$, which behaves exactly like the random ensemble).
The following lemma provides a useful bound by utilizing de-Caen's lower bound on the probability of a union of events \cite{j:deCaenALowerBound} instead.
\begin{lem}
\label{t:lem_pr_lin_error_m_x_y}
Let $\mbf{x}_m$ be some codeword from the linear ensemble such that $q^K>2$ (corresponding to message vector $\mbf{u}_m$), let $\mbf{y}$ be some channel output, and let $\rho$ be a constant such that $\rho\geq1$, then the ensemble average of the probability of the union of all pairwise error events $m'\neq m$ can be upper and lower bounded as
\begin{equation}
\label{e:lem_pr_lin_error_m_x_y}
\frac{(M-1)\alpha}{q} - [(M-1)\alpha]^\rho \leq \Pr\left(\bigcup_{m'\neq m}A_{m'}(\mbf{x}_m,\mbf{y})\right) \leq (M-1)\alpha
\end{equation}
where $\alpha\triangleq\Pr(A_{m'}(\mbf{x}_m,\mbf{y}))$.
\end{lem}

\begin{proof}
In our settings, de-Caen's lower bound on the probability of a union of events can be expressed as
\begin{align}
\label{e:def_de_caen_0}
\Pr& \left( \bigcup_{m'\neq m} A_{m'}(\mbf{x}_m,\mbf{y}) \right) \nonumber \\
&\geq \sum_{m'\neq m} \frac{[\Pr(A_{m'}(\mbf{x}_m,\mbf{y}))]^2}{\sum_{m''\neq m} \Pr(A_{m'}(\mbf{x}_m,\mbf{y}) \cap A_{m''}(\mbf{x}_m,\mbf{y}))}.
\end{align}
To simplify, let us rewrite \eqref{e:def_de_caen_0} as
\begin{align}
\label{e:def_de_caen_1}
\Pr& \left( \bigcup_{m'\neq m} A_{m'}(\mbf{x}_m,\mbf{y}) \right) \nonumber \\
&\geq \sum_{m'\neq m} \frac{\Pr(A_{m'}(\mbf{x}_m,\mbf{y}))}{\sum_{m''\neq m} \frac{\Pr(A_{m'}(\mbf{x}_m,\mbf{y}) \cap A_{m''}(\mbf{x}_m,\mbf{y}))}{\Pr(A_{m'}(\mbf{x}_m,\mbf{y}))}}.
\end{align}
Beginning with the denominator of \eqref{e:def_de_caen_1}
\begin{align}
\label{e:def_de_caen_partial}
\sum_{m''\neq m}& \frac{\Pr(A_{m'}(\mbf{x}_m,\mbf{y}) \cap A_{m''}(\mbf{x}_m,\mbf{y}))}{\Pr(A_{m'}(\mbf{x}_m,\mbf{y}))} \nonumber \\
&= 1 + \sum_{m''\neq m,m'} \frac{\Pr(A_{m'}(\mbf{x}_m,\mbf{y}) \cap A_{m''}(\mbf{x}_m,\mbf{y}))}{\Pr(A_{m'}(\mbf{x}_m,\mbf{y}))} \nonumber \\
&= 1 + \alpha^{-1} \sum_{m''\neq m,m'} \Pr(A_{m'}(\mbf{x}_m,\mbf{y}) \cap A_{m''}(\mbf{x}_m,\mbf{y})) \nonumber \\
&\leq 1 + (q^K-q)\alpha + (q-2)
\end{align}
where the inequality is achieved by breaking up the $q^K-2$ sum indexes ($m''\neq m,m'$) into two groups based on whether $\mbf{u}_m\in\textnormal{span}^*(\mbf{u}_{m'},\mbf{u}_{m''})$ or not, and then using \eqref{e:lem_pr_lin_double_error_event_1} and \eqref{e:lem_pr_lin_double_error_event_2}.
Plugging \eqref{e:def_de_caen_partial} into \eqref{e:def_de_caen_1} results in
\begin{align}
\label{e:lem_pr_lin_error_m_x_y_down_1}
\Pr \left( \bigcup_{m'\neq m} A_{m'}(\mbf{x}_m,\mbf{y}) \right)
&\geq \frac{(q^K-1)\alpha}{(q^K-q)\alpha + (q-1)} \nonumber \\
&\geq \frac{(q^K-1)\alpha}{(q^K-1)\alpha + (q-1)}.
\end{align}
We continue to analyze the last inequality of \eqref{e:lem_pr_lin_error_m_x_y_down_1} for two cases: \\
When $(q^K-1)\alpha<1$,
\begin{align}
\label{e:lem_pr_lin_error_m_x_y_down_2}
\frac{(q^K-1)\alpha}{(q^K-1)\alpha + (q-1)}
&\geq \frac{(q^K-1)\alpha}{1+(q-1)} \nonumber \\
&\geq \frac{(q^K-1)\alpha}{q} - [(q^K-1)\alpha]^\rho
\end{align}
for any constant $\rho\geq1$.
When $(q^K-1)\alpha\geq 1$, the last inequality of \eqref{e:lem_pr_lin_error_m_x_y_down_2} is always negative and is thus a trivial lower bound.

Finally, plugging \eqref{e:lem_pr_lin_error_m_x_y_down_2} back into \eqref{e:lem_pr_lin_error_m_x_y_down_1} for the lower bound and using the union bound as the upper bound, while replacing $q^K$ with $M$, we arrive at \eqref{e:lem_pr_lin_error_m_x_y}, which can be regarded as the linear ensemble alternative to \eqref{e:pr_rnd_error_m_x_y}.
\end{proof}

\subsection{Bounding $\overline{P}_{e,m}$}
\begin{thm}
\label{t:thm_lin_tight}
Let $\overline{P}_{e,m}(N)$ be the linear ensemble's average error probability upon transmission of a codeword corresponding to the $m$'th message, where $m\in[0,M-1]$, over a channel from the class $\mathcal{U}$.
Then, $\overline{P}_{e,m}(N)$ is independent of $m$, and for any $R<C$
\begin{equation}
-\lim_{N\rightarrow\infty} \frac{\log(\overline{P}_{e,m}(N))}{N} = E_r(R)
\end{equation}
where $N$, $R$, $C$ and $E_r(R)$ are the code's dimension, rate, the channel's capacity, and Gallager's random-coding error exponent \cite{b:GallagerInformationTheory}, respectively.
\end{thm}

\begin{proof}
Using \eqref{e:lem_pr_lin_error_m_x_y} and \eqref{e:pr_error_m_x_y}, \eqref{e:pe_ensemble} can be bounded as follows:
\begin{equation}
\label{e:pe_lin_bounds_1}
\frac{P_1}{q}-P_2 \leq \overline{P}_{e,m} \leq P_1
\end{equation}
where
\begin{align}
&P_1 = (M-1)\sum_{\mbf{x}_m} \sum_\mbf{y} Q_N(\mbf{x}_m) P_N(\mbf{y}|\mbf{x}_m) \Pr(A_{m'}(\mbf{x}_m,\mbf{y})) \nonumber \\
&P_2 = (M-1)^{\rho}\sum_{\mbf{x}_m} \sum_\mbf{y} Q_N(\mbf{x}_m) P_N(\mbf{y}|\mbf{x}_m) [\Pr(A_{m'}(\mbf{x}_m,\mbf{y}))]^{\rho}
\end{align}
and $\rho\geq1$.\footnotemark{}
\footnotetext{Maximizing $P_2$ over $\rho\geq1$ for $R<R_{\textnormal{cr}}$ results in the sphere packing exponent.}

In order to complete the proof, we need to show that the left-hand-side of \eqref{e:pe_lin_bounds_1} can be equated as
\begin{equation}
\label{e:pe_lin_bounds_2}
\frac{P_1}{q}-P_2 \doteq P_1
\end{equation}
where $\doteq$ implies asymptotic exponential equality in $N$, and where the dependence on $N$ is suppressed.
Firstly note that when $\rho$ is constrained to $\rho>0$, $P_2$ is Gallager's sphere packing bound \cite{b:GallagerInformationTheory}.
As such there exists a critical rate $R_{\textnormal{cr}}$, such that $P_2$ is maximized by some $\rho>1$ for all $R<R_{\textnormal{cr}}$.
This immediately implies that for those rates $P_2$ goes to zero with a larger exponent in $N$ than $P_1$, thus
\begin{equation}
\label{e:pe_lin_bounds_3}
\frac{P_1}{q}-P_2 \doteq \frac{P_1}{q}.
\end{equation}
The remainder is simple since
\begin{equation}
\label{e:pe_lin_bounds_4}
\frac{q^{NR}}{q} = q^{NR-1} = q^{N\left(R-\frac{1}{N}\right)} \doteq q^{NR},
\end{equation}
thus
\begin{equation}
\label{e:pe_lin_bounds_5}
\frac{P_1}{q} \doteq P_1.
\end{equation}

$R_{\textnormal{cr}}$ defined above is exactly the critical rate of the channel.
As such, it is well known that the random-coding exponent $E_r(R)$ is tight for rates $R\geq R_{\textnormal{cr}}$ \cite{b:GallagerInformationTheory}.
For rates $R<R_{\textnormal{cr}}$, $E_r(R)$ is given by
\begin{equation}
E_r(R) = - \lim_{N\rightarrow\infty} \frac{\log(P_1)}{N}
\end{equation}
and its tightness is thus implied by the correctness of \eqref{e:pe_lin_bounds_2}.
\end{proof}

\section{Extension for Lattices}
We proceed to show that our results extend for a random ensemble of unbounded lattices.
Specifically, let us show that Poltyrev's random-coding exponent for lattices \cite{j:PoltyrevOnCoding} is tight for Loeliger's ensemble \cite{j:LoeligerAveragingBounds}, below the critical Normalized-Log-Density (NLD) $\delta_{\textnormal{cr}}$.\footnotemark{}
\footnotetext{Poltyrev's random-coding exponent above $\delta_{\textnormal{cr}}$ coincides with the sphere lower bound exponent. See \cite{j:PoltyrevOnCoding}.}

\subsection{The Channel}
Our proof is constructed for the additive noise channel
\begin{equation}
\label{e:lattice_channel}
\mbf{y} = \mbf{x}+\mbf{z}
\end{equation}
where $\mbf{y}$ is the channel output, $\mbf{x}$ is the channel input, and $\mbf{z}$ is the AWGN independent of $\mbf{x}$.
In order to simplify the notation, the \pdf of $\mbf{z}$ is denoted by $f(\mbf{z})$ rather than explicitly.\footnotemark{}
\footnotetext{Alternatively, $\mbf{z}$ can be taken to be distributed such that its \CDF is continuous and its \pdf is isotropic and monotonically non-increasing in $\norm{\mbf{z}}$.}

\subsection{Loeliger's Ensemble}
A lattice in Loeliger's ensemble is constructed by taking a linear $(N,K,q)$ code as defined in \eqref{e:def_lin_code} with $\mbf{v}=0$, scaling it by a constant $\beta$ per dimension, and tiling $\mathbb{R}^N$ with it by \textit{construction-A} \cite{b:ConwaySpherePackings} (i.e. $\Lambda=\{\lambda:\lambda\mod (\beta q)^N \in \mathcal{C}_{\beta}\}$ where $\mathcal{C}_{\beta}$ is the scaled code).
The constant $\beta$ is selected such that the lattice density $\gamma$ is constant, thus,
\begin{align}
\label{e:def_lat_gamma}
&\gamma = \frac{q^K}{(\beta q)^N} \nonumber \\
&\Rightarrow \beta = q^{\frac{K-N}{N}} \gamma^{-\frac{1}{N}}.
\end{align}
Loeliger's ensemble for some selection of $(N,K,q,\gamma)$ is a uniformly distributed set of lattices, constructed by extending the codes of the $(N,K,q,\mbf{v}=0)$ linear ensemble to lattices, as discussed above.
Finally define Loeliger's asymptotic ensemble for some selection of $(N,K,\gamma)$ as the limit of Loeliger's $(N,K,q,\gamma)$ ensemble for $q\rightarrow\infty$.
Specifically, our proof is constructed for Loeliger's asymptotic ensemble.

\subsection{Poltyrev's Random-Coding Exponent}
A widely accepted framework for lattice codes' error analysis is commonly referred to as Poltyrev's setting \cite{j:PoltyrevOnCoding}.
In Poltyrev's setting the code's shaping region, defined as the finite subset of the otherwise infinite set of lattice points, is ignored, and the lattice structure is analyzed for its coding (soft packing) properties only.
Consequently, the usual rate variable $R$ is infinite and replaced by the NLD, $\delta=\frac{\log(\gamma)}{N}$.

The average error probability for a \textit{uniformly distributed}\footnotemark{} ensemble of lattices transmitted over an AWGN channel with noise variance $\sigma^2$, can be expressed in the following exponential form \cite{j:PoltyrevOnCoding}, \cite{j:IngberFiniteDimensional}
\footnotetext{Such as Loeliger's asymptotic ensemble or Minkowski-Hlawka-Siegel \cite{j:SiegelAMeanValue,j:MacbeathAModifiedForm}.}
\begin{equation}
\label{eq_bound_mhs_exp}
\overline{P}_e
\leq e^{-N(E_{\textnormal{r}}(\delta)+o(1))}
\end{equation}
with
\begin{equation}
E_{\textnormal{r}}(\delta)
= \left\{
\begin{array}{ll}
(\delta^*-\delta) + \log{\frac{e}{4}},                                         & \delta\leq\delta_{\textnormal{cr}} \\
\frac{e^{2(\delta^*-\delta)}-2(\delta^*-\delta)-1}{2},
& \delta_{\textnormal{cr}}\leq\delta<\delta^* \\
0,                                                                             & \delta\geq\delta^*
\end{array} \right.
\end{equation}
\begin{align}
\delta^* &= \frac{1}{2}\log{\frac{1}{2\pi e\sigma^2}} \\
\delta_{\textnormal{cr}} &= \frac{1}{2}\log{\frac{1}{4\pi e\sigma^2}}
\end{align}
where $o(1)$ goes to zero asymptotically with $N$.

\subsection{The Bounding Method}
The error probability for our channel \eqref{e:lattice_channel} can be trivially bounded as
\begin{equation}
\label{e:pe_lat_bounds_0}
\Pr(e, \norm{\mbf{z}}\leq r) \leq \overline{P}_{e,m} \leq \Pr(e, \norm{\mbf{z}}\leq r)+\Pr(\norm{\mbf{z}}> r)
\end{equation}
where $e$ is the maximum-likelihood decoding error event, and $r$ is an optimization parameter.
Due to the code's linearity and the additive noise's independence, we can assume, with no loss of generality, that index $m=0$ (i.e. $\lambda=\lambda_0$) is transmitted and analyze $\overline{P}_{e,m=0}$.
The leftmost term in the above inequality can then be expressed as
\begin{equation}
\label{e:pe_lat_bounds_1}
\Pr(e, \norm{\mbf{z}}\leq r)
= \Pr\left( \bigcup_{\lambda\neq\lambda_0} e_{\lambda} , \norm{\mbf{z}}\leq r \right)
\end{equation}
where $e_{\lambda}$ is the event that $\Pr(\lambda|\mbf{y},m=0)\geq\Pr(\lambda_0|\mbf{y},m=0)$.
The right-hand-side of \eqref{e:pe_lat_bounds_1} can be upper bounded by the union-bound as
\begin{equation}
\Pr\left( \bigcup_{\lambda\neq\lambda_0} e_{\lambda} , \norm{\mbf{z}}\leq r \right)
\leq \sum_{\lambda\neq\lambda_0} \Pr(e_{\lambda} , \norm{\mbf{z}}\leq r).
\end{equation}
Plugging the union-bound into the right-hand-side of \eqref{e:pe_lat_bounds_0} and optimizing for $r$ results in
\begin{equation}
\label{e:pe_lat_bounds_2}
\overline{P}_{e,m=0} \leq \sum_{\lambda\neq\lambda_0} \Pr(e_{\lambda},\norm{\mbf{z}}\leq r^*)+\Pr(\norm{\mbf{z}}>r^*)
\end{equation}
where $r^*$ is selected to minimize the expression.
Plugging $r^*$ into the left-hand-side of \eqref{e:pe_lat_bounds_0} results in
\begin{equation}
\label{e:pe_lat_bounds_3}
\overline{P}_{e,m=0} \geq \Pr(e,\norm{\mbf{z}}\leq r^*).
\end{equation}
The goal of our proof is to show that for NLDs $\delta: \delta<\delta_{\textnormal{cr}}$ the upper \eqref{e:pe_lat_bounds_2} and lower \eqref{e:pe_lat_bounds_3} bounds
asymptotically exponentially coincide.
Most of our development focuses on the analysis of the expression
\begin{equation}
\label{e:pe_lat_bounds_4}
\Pr(e,\norm{\mbf{z}}\leq r^*)
\end{equation}
which we refer to as the union of error events.

\subsection{Analysis of the lattice ensemble via a linear ensemble}
The first step is to set the alphabet size $q$ of the ensemble to be large enough, such that all lattice points, relevant to the noise region governed by $\norm{\mbf{z}}\leq r^*$, are contained in a single code cube centered at the origin.
This, as is clarified shortly, enables analyzing the lattice ensemble using tools previously designed for the linear ensemble.
For a start, consider a specific lattice, and examine the expression
\begin{equation}
\Pr\left( \bigcup_{\lambda\neq\lambda_0} e_{\lambda} , \norm{\mbf{z}}\leq r^* \right).
\end{equation}
For our channel the probability of the pairwise error event $e_{\lambda}$ is non-zero only for lattice points $\lambda: \lambda\in\Ball(\mbf{z},\norm{\mbf{z}})$, where $\Ball(\mbf{z},\norm{\mbf{z}})$ is a Euclidean ball of radius $\norm{\mbf{z}}$ centered at $\mbf{z}$.
This together with $\norm{\mbf{z}}\leq r^*$, zeroes the probability of the pairwise error event for any lattice point with norm $\norm{\lambda}>2r^*$.
The condition on $q$ that forces all lattice points $\lambda: \norm{\lambda}\leq 2r^*$ to be contained in a single code cube centered at the origin\footnotemark{} is
\footnotetext{The Voronoi cell of the lattice is always contained in the centered code cube. See \cite{b:ZamirBook}.}
\begin{align}
\label{e:cond_lat_rstar}
&2r^*<\frac{\beta q}{2} \nonumber \\
&\Rightarrow r^*< 0.25\cdot q^{\frac{K}{N}} \gamma^{-\frac{1}{N}} \nonumber \\
&\Rightarrow q>(4r^*)^\frac{N}{K} \gamma^{\frac{1}{K}}.
\end{align}
The \textit{construction-A} of Loeliger's ensemble completely tiles the $N$-dimensional space with non-overlapping copies of the fundamental linear code.
Before scaling by $\beta$, the fundamental linear code is contained in the $[0,q-1]^N$ cube.
Clearly, the construction does not consist of a centered code cube.
One method to describe the centered code cube is by the following one-to-one re-mapping operation: Take a lattice point from the fundamental code cube and subtract $q$ from its dimensions that exceed $\frac{q}{2}$.
This re-mapping produces a one-to-one mapping from any message index $m'\in[0,M-1]$ to its corresponding codeword in the centered code cube.
A two dimensional illustration of the index re-mapping is depicted in Figure \ref{fig_const_a_remap}.
\input{const_a_remap.tpx}
With $q$ obeying \eqref{e:cond_lat_rstar}, we can say that
\begin{equation}
\label{e:cond_lat_to_lin}
\Pr\left( \bigcup_{\lambda\neq\lambda_0} e_{\lambda} , \norm{\mbf{z}}\leq r^* \right)
= \Pr\left( \bigcup_{m'\in[1,M-1]} e_{\lambda_{m'}} , \norm{\mbf{z}}\leq r^* \right)
\end{equation}
where $\lambda_{m'}$ corresponds to the codewords of the centered code cube.

Let us now switch back to the ensemble (rather than a specific lattice).
One should note, that mapping the lattice points belonging to the centered code cube to the indexes $m'$, enables analyzing them as the codewords of the underling linear code.
As such, one can say that they are distributed according to Theorem \ref{t:thm_pr_lin_1_cond_k} for the case where $\mbf{v}=0$ and $m=0$.
The distribution of lattice-point $\lambda_{m'}$ conditioned on $\lambda_0$ is thus
\begin{equation}
Q_N(\lambda_{m'}|\lambda_0) = Q_N(\lambda_{m'}) = q^{-N}.
\end{equation}
The distribution of lattice-point $\lambda_{m''}$ conditioned on $\lambda_{m'}$ and $\lambda_0$ is
\begin{equation}
Q_N(\lambda_{m''}|\lambda_{m'},\lambda_0) = Q_N(\lambda_{m''}|\lambda_{m'}) = q^{-N}
\end{equation}
when $\mbf{u}_{m''}\notin \spann(\mbf{u}_{m'})$, and
\begin{equation}
Q_N(\lambda_{m''}|\lambda_{m'})
= \left\{
\begin{array}{ll}
1 & : \lambda_{m''} = \mathcal{L}(\lambda_{m'}) \\
0 & : \lambda_{m''} \neq \mathcal{L}(\lambda_{m'})
\end{array}
\right.
\end{equation}
when $\mbf{u}_{m''}\in \spann(\mbf{u}_{m'})$, where $\mathcal{L}$ is any linear transformation such that $\mbf{u}_{m''} = \mathcal{L}(\mbf{u}_{m'})$.
One should note, that both $\spann(\cdot)$ and $\mathcal{L}(\cdot)$ are taken with respect to $F_q^N$.

\subsection{Bounding the union of error events}
\begin{lem}
\label{l:lem_pr_lat_error_union}
Let the zero'th lattice point $\lambda_0$ from Loeliger's ensemble, that obeys \eqref{e:cond_lat_rstar}, be transmitted over an AWGN channel with noise distribution given by $f(\mbf{z})$, and let $r^*$ be defined by \eqref{e:pe_lat_bounds_2}, then the ensemble's average of the joint probability of the union of all pairwise error events $\lambda\neq\lambda_0$ and $\norm{\mbf{z}}\leq r^*$ can be upper and lower bounded as
\begin{equation}
\label{e:lem_pr_lat_error_union}
\frac{(M-1)}{q} \int_{\norm{\mbf{z}}\leq r^*} \alpha \cdot f(\mbf{z}) d\mbf{z}
\leq \Pr(e, \norm{\mbf{z}}\leq r^*)
\leq (M-1) \int_{\norm{\mbf{z}}\leq r^*} \alpha \cdot f(\mbf{z}) d\mbf{z}
\end{equation}
where $M=q^K$, $\alpha\triangleq\Pr(A_{m'}(\mbf{z})|\mbf{z})$, and $A_{m'}(\mbf{z})$ is the event that an index $m'\neq0$ is such as to cause an error for $m=0$ transmitted and $\mbf{y}=\mbf{z}$ received.
\end{lem}

\begin{proof}
Define the pairwise error event $A_{m'}(\mbf{z})$ as the event that an index $m'\neq0$ is such as to cause an error for $m=0$ transmitted and $\mbf{y}=\mbf{z}$ received, and proceed to calculate two expressions.
The first is the probability of the event $A_{m'}(\mbf{z})$ conditioned on $\mbf{z}$
\begin{align}
\label{e:pr_lat_single_error_event}
\Pr(A_{m'}(\mbf{z})|\mbf{z})
&= \sum_{\lambda_{m'}\in\beta\mathbb{Z}_q^N} Q_N(\lambda_{m'}) \mathds{1}(\lambda_{m'}\in\Ball(\mbf{z},\norm{\mbf{z}})|\mbf{z}) \nonumber \\
&= q^{-N} \sum_{\lambda_{m'}\in\beta\mathbb{Z}_q^N} \mathds{1}(\lambda_{m'}\in\Ball(\mbf{z},\norm{\mbf{z}})|\mbf{z}).
\end{align}
Clearly, $\Pr(A_{m'}(\mbf{z})|\mbf{z})$ is independent of $m'$.
Temporarily suppressing $\mbf{z}$, denote
\begin{equation}
\Pr(A_{m'}(\mbf{z})|\mbf{z}) \triangleq \alpha.
\end{equation}
From Lemma \ref{t:lem_lin_double_error_event}, the probability of the intersection of two error events, conditioned on $\mbf{z}$ is
\begin{equation}
\Pr(A_{m'}(\mbf{z})\cap A_{m''}(\mbf{z})|\mbf{z}) = \alpha^2
\end{equation}
when $\mbf{u}_{m''}\notin \spann(\mbf{u}_{m'})$, and
\begin{equation}
\Pr(A_{m'}(\mbf{z})\cap A_{m''}(\mbf{z})|\mbf{z})\leq \alpha
\end{equation}
when $\mbf{u}_{m''}\in \spann(\mbf{u}_{m'})$. \\
Using \eqref{e:cond_lat_to_lin}, the error union probability \eqref{e:pe_lat_bounds_4} can be restated as
\begin{align}
\label{e:pr_lat_error_event}
\Pr(e, \norm{\mbf{z}}\leq r^*)
&= \Pr\left( \bigcup_{m'\in[1,M-1]} A_{m'}(\mbf{z}) , \norm{\mbf{z}}\leq r^* \right) \nonumber \\
&= \int_{\norm{\mbf{z}}\leq r^*} f(\mbf{z}) \Pr\left( \left. \bigcup_{m'\in[1,M-1]} A_{m'}(\mbf{z}) \right|\mbf{z} \right) d\mbf{z}
\end{align}
where $f(\mbf{z})$ is the \pdf of $\mbf{z}$.
The internal probability term of the right-hand-side of \eqref{e:pr_lat_error_event} can be upper bounded by the union bound as
\begin{equation}
\Pr\left( \left. \bigcup_{m'\in[1,M-1]} A_{m'}(\mbf{z}) \right|\mbf{z} \right)
\leq (M-1)\alpha
\end{equation}
and lower bounded by the de-Caen inequality as
\begin{equation}
\Pr\left( \left. \bigcup_{m'\in[1,M-1]} A_{m'}(\mbf{z}) \right|\mbf{z} \right)
\geq \frac{(M-1)\alpha}{(M-1)\alpha + (q-1)}
\geq \frac{(M-1)\alpha}{q}
\end{equation}
where the first inequality follows from \eqref{e:lem_pr_lin_error_m_x_y_down_1}, and the second follows from the definition of $r^*$ \eqref{e:pe_lat_bounds_2} since $(M-1)\alpha\leq1$ for any $\mbf{z}:\norm{\mbf{z}}\leq r^*$.\footnotemark{}
\footnotetext{$\int_{\norm{\mbf{z}}\leq r^*} (M-1)\alpha\cdot f(\mbf{z})d\mbf{z}+\int_{\norm{\mbf{z}}>r^*} f(\mbf{z})d\mbf{z} = \int \min\{(M-1)\alpha,1\}\cdot f(\mbf{z})d\mbf{z}$}

Finally, plugging the above back into \eqref{e:pr_lat_error_event} completes the proof.
\end{proof}

\subsection{Bounding $\overline{P}_{e,m=0}$}
\begin{thm}
\label{t:thm_lat_tight}
Consider Loeliger's ensemble where the linear code's cardinality $q$, the dimension $N$, and the rate $R$ obey the relationship
\begin{equation}
\label{e:pe_lat_bounds_10}
O(0.5R^{-1}\log(N))\leq \log(q)\leq o(N),
\end{equation}
and let $\overline{P}_{e,m=0}(N)$ be the ensemble's average error probability upon transmission of the zero'th lattice point over an AWGN channel. Then, $\overline{P}_{e,m=0}(N)$ is independent of $m$, and for any $\delta<\delta^*$
\begin{equation}
-\lim_{N\rightarrow\infty} \frac{\log(\overline{P}_{e,m=0}(N))}{N} = E_r(\delta)
\end{equation}
where $\delta$, $\delta^*$, and $E_r(\delta)$ are the NLD, the maximum achievable NLD, and Plotyrev's random-coding error exponent \cite{j:PoltyrevOnCoding}, respectively.
\end{thm}

\begin{proof}
Similarly to the treatment of linear codes, note from \eqref{e:pe_lat_bounds_0} and \eqref{e:lem_pr_lat_error_union} that $\overline{P}_{e,m=0}$ can be bounded as follows:
\begin{equation}
\label{e:pe_lat_bounds_5}
\frac{P_1}{q} \leq \overline{P}_{e,m=0} \leq P_1+P_2
\end{equation}
where
\begin{align}
&P_1 = (M-1) \int_{\norm{\mbf{z}}\leq r^*} \alpha \cdot f(\mbf{z}) d\mbf{z} \nonumber \\
&P_2 = \int_{\norm{\mbf{z}}>r^*} f(\mbf{z})d\mbf{z}.
\end{align}
Recall that asymptotically in $q$ (for any $N>1$) \cite{j:IngberFiniteDimensional,j:LoeligerAveragingBounds}
\begin{align}
\label{e:pe_lat_bounds_6}
&P_1 \xrightarrow{\scriptscriptstyle q\to\infty} e^{N\delta} V_N \int_0^{r^*} f_{\norm{\mbf{z}}}(\rho) d\rho \nonumber \\
&r^* \xrightarrow{\scriptscriptstyle q\to\infty} e^{-\delta} V_N^{-1/N}
\end{align}
where $\delta$ is the lattice's NLD (i.e. $\gamma=e^{N\delta}$), $V_N=\frac{\pi^{N/2}}{\Gamma(N/2+1)}$ is the volume of an $N$-dimensional unit sphere, and $f_{\norm{\mbf{z}}}(\rho)$ is the \pdf of the noise magnitude $\norm{\mbf{z}}$.

In order to complete the proof, we need to show that the left and right-hand-sides of \eqref{e:pe_lat_bounds_5} can be exponentially equated in a similar fashion to the linear ensemble case,
\begin{equation}
\label{e:pe_lat_bounds_7}
\frac{P_1}{q} \doteq P_1+P_2
\end{equation}
where $\doteq$ implies asymptotic equality, simultaneously approached in $q$ and $N$, and where the dependence on $q$ and $N$ is suppressed.
Increasing $q$ is required for the ensemble to be sufficiently dense, such that \eqref{e:pe_lat_bounds_6} holds, while increasing $N$ is necessary for the exponent.
Due to the ensemble's construction, it is necessary to define the relationship between $q$, $K$, and $N$.
Let us begin by defining a relationship and end by showing that such a relationship suffices to achieve \eqref{e:pe_lat_bounds_7}.
Restrict the ratio between $K$ and $N$ to be approximately constant.
One way to achieve this is by selecting
\begin{equation}
\label{e:pe_lat_bounds_8}
K = \lceil RN \rceil
\end{equation}
where $R$ is constant, and $\lceil \cdot \rceil$ denotes the nearest-integer ceil operator.
An asymptotic lower bound on $q$ can be found by plugging \eqref{e:pe_lat_bounds_6} into \eqref{e:cond_lat_rstar}
\begin{equation}
\label{e:pe_lat_bounds_9}
\log(q)> \frac{N}{K}\log(4)-K^{-1}\log(V_N).
\end{equation}
Using $K$ as defined by \eqref{e:pe_lat_bounds_8}, the right-hand-side of \eqref{e:pe_lat_bounds_9} can be upper bounded as
\begin{align}
\frac{N}{K}\log(4)-K^{-1}\log(V_N)
&\leq R^{-1}(\log(4)-N^{-1}\log(V_N)) \nonumber \\
&< 0.5R^{-1}\log(N) - 0.5R^{-1}\log(2\pi) + R^{-1}\log(4) \nonumber \\
&= O(0.5R^{-1}\log(N))
\end{align}
where the second inequality follows from the definition of $V_N$.
Choosing $q$ such that
\begin{equation}
\label{e:pe_lat_bounds_10}
O(0.5R^{-1}\log(N))\leq \log(q)\leq o(N)
\end{equation}
obviously obeys \eqref{e:pe_lat_bounds_9}.
The upper bound in \eqref{e:pe_lat_bounds_10} is clarified shortly.\footnotemark{}
\footnotetext{Proving the existence of lattices that are simultaneously good for both coding and quantization, by similar constructions, leads to similar requirements on the order of the limits \cite{p:OrdentlichASimpleProof}.}
From \eqref{e:pe_lat_bounds_6} we can say that asymptotically in $q$
\begin{equation}
\label{e:pe_lat_bounds_11}
\frac{P_1}{q} \xrightarrow{\scriptscriptstyle q\to\infty} e^{N(\delta-N^{-1}\log(q))} V_N \int_0^{r^*} f_{\norm{\mbf{z}}}(\rho) d\rho.
\end{equation}
From \eqref{e:pe_lat_bounds_11} and the upper bound set in \eqref{e:pe_lat_bounds_10}
\begin{equation}
\frac{P_1}{q} \doteq P_1
\end{equation}
asymptotically in $N$.
Furthermore, it is known that for NLDs below $\delta_{\textnormal{cr}}$, $P_2$ goes to zero with a larger exponent in $N$ than $P_1$ \cite{j:PoltyrevOnCoding,j:IngberFiniteDimensional} so that
\begin{equation}
P_1 \doteq P_1+P_2.
\end{equation}

This shows that for NLDs $\delta: \delta<\delta_{\textnormal{cr}}$, Poltyrev's random-coding bound is exponentially tight for Loeliger's asymptotic ensemble.
This, together with the tightness of the Poltyrev's random-coding bound for NLDs $\delta: \delta\geq\delta_{\textnormal{cr}}$ \cite{j:PoltyrevOnCoding}, leads to the conclusion that Poltyrev's random-coding bound is exponentially tight for Loeliger's asymptotic ensemble at all NLDs.
Plugging the Guassian \pdf for $f(\mbf{z})$ and simplifying the expression $P_1+P_2$ (asymptotically in $q$ and $N$, see \cite{j:IngberFiniteDimensional}) completes the proof.
\end{proof}

\section{Concluding Remarks}
The tightness of the random-coding exponent over the random linear ensemble has long been conjectured.
The main contribution made by this paper is a rigorous proof showing that this conjecture is indeed true.
An extension of the proof to the random lattice ensemble is also shown.
A secondary, but perhaps significant, contribution is the explicit distribution of codewords conditioned on other codewords for the random linear ensemble, which may prove useful elsewhere.

\section*{Acknowledgement}
The authors would like to thank Alexander Barg for helpful discussions.

\bibliographystyle{IEEEtran}
\bibliography{IEEEabrv,Yuval}

\end{document}